\documentclass[3p,reviewcopy,11pt]{elsarticle}
\usepackage{ecrc}
\usepackage{color}
\usepackage{pifont}
\usepackage{amsmath, amsthm, amssymb, amsfonts}
\usepackage{graphicx}
\usepackage{geometry}
\usepackage{txfonts}
\usepackage{natbib}
\usepackage[T1]{fontenc}
\usepackage[all]{xy}
\usepackage{times}
\usepackage{yfonts}
\usepackage{soul}

\theoremstyle{plain}
\newtheorem{Theorem}{Theorem}

\newtheorem{Lemma}{Lemma}

\newtheorem{Proposition}{Proposition}
\newtheorem{Corollary}{Corollary}

\theoremstyle{definition}
\newtheorem{Definition}{Definition}
\theoremstyle{remark}
\newtheorem{Remark}{Remark}
\newtheorem*{cou}{Counterexample}

\volume{??}

\firstpage{1} \journal{Discrete Mathematics}

\usepackage{amssymb}

\usepackage[figuresright]{rotating}

\begin{document}

\begin{frontmatter}



\title{On Counting Subring-Subcodes of Free Linear Codes Over Finite Principal Ideal Rings}


\author[rkb]{Ramakrishna Bandi}\ead{bandi.ramakrishna@gmail.com}
\author[aft]{Alexandre Fotue Tabue}\ead{alexfotue@gmail.com}
\author[ed]{Edgar Martínez-Moro\corref{cor1}} \ead{edgar.martinez@uva.es}\cortext[cor1]{Corresponding author. Partially
funded by MINECO MTM2015-65764-C3-1-P grant.}

\address[rkb]{Department of Mathematics, School of Engineering and Technology, BML Munjal University, India}
\address[aft]{Research and Training Unit for Doctorate in Mathematics, Computer Sciences and Applications, University of Yaounde I}
\address[ed]{Institute of Mathematics, University of Valladolid, Spain}

\begin{abstract} Let $\texttt{R}$ be a finite principal ideal ring and $\texttt{S}$ the Galois extension of $\texttt{R}$ of degree $m.$ For $k$ and $k'$,
positive integers we determine the number of free
$\texttt{S}$-linear codes $\mathcal{B}$ of length $\ell$ with the
property $k=\texttt{rank}_\texttt{S}(\mathcal{B})$ and
$k'=\texttt{rank}_\texttt{R}(\mathcal{B}\cap\texttt{R}^\ell)$. This corrects a wrong result \cite[Theorem 6]{Lyle78} which was given in the case of finite fields.
\end{abstract}


\end{frontmatter}

\section{Introduction}

Many codes over a finite field $\mathbb F$ can be seen as
subfield-subcodes of codes that are defined over a some extension field of $\mathbb F$.
Unfortunately there is no general formula for the dimension or the
minimum distance of such codes. However there are good bounds for some families of
codes and also some  formulae have been obtained for the true dimension
in the case of some alternant codes, toric codes or some families of
Goppa codes (see \cite{hernando} and the references therein). In
a more general way,  given   a finite principal
ideal ring $\texttt{R}$ and a Galois extension $\texttt{S}$
 of $\texttt{R}$ of degree $m$,  one can define
an $\texttt{S}$-linear code $\mathcal{B}$ of length $\ell$ as  $\texttt{S}$-module of $\texttt{S}^\ell$. The rank of a linear code
$\texttt{rank}_\texttt{S}(\mathcal{B})$ will be the minimum number of
its generators.  The $\texttt{R}$-linear code $\mathcal{B}\cap\texttt{R}^\ell$ is called the
\emph{subring-subcode} of $\mathcal{B}$ to $\texttt{R}.$ The relationship between subring subcodes and trace codes over finite chain rings has been studied in
\cite{MNR13} revealing a generalization of the classical Delsarte's Theorem.
 In \cite{Lyle78}, Lyle determined the  number of distinct $\mathbb{F}_{q^m}$-linear codes of length $\ell$ and of dimension $k,$ with a fixed dimension of their subfield subcodes over $\mathbb{F}_q$. However the result is incorrect. The following example illustrates the same.

\begin{cou}
We have 21  different  2-dimensional codes of
length 3 over $\mathbb F_4$. Thus
from  \cite[Theorem 6]{Lyle78}, it follows that the $\mathbb
F_2$-subcode of any of them is the trivial code
$\{ (0,0,0)\}.$ Consider now
$\mathbb F_4=\{ 0,1,\alpha ,\beta\}$ such that $\alpha \cdot \beta
=\alpha + \beta =1$, and the 2-dimensional code
 $\mathcal C$ generated by $(1,0,\alpha)$ and $(0,1,\beta)$ .
Clearly $(1,1,1)= (1,0,\alpha)+(0,1,\beta)$ belongs to $\mathcal C\cap \mathbb F_2^3$, which is a contradiction to the previous statement.
\end{cou}

In this paper, we present a correct answer to the counting problem in
Lyle's paper in the broader class of codes possible where linear
algebra can be  accomplished, that
is free codes over finite principal ideal rings that indeed also
covers the finite field case. The Galois invariance of a code
\cite{MNR13,Fotue} plays an important role on the paper.
 The paper is organized as follows: In Section 2 we present some definitions and
 preliminary results. Section 3 provides a formula for counting
the number  of free $\texttt{S}$-linear codes of given rank whose
the rank of its subring subcode is given. Finally in Section 4 we
extend the results to PIRs.

\section{Preliminaries}


Throughout this section $\texttt{R}$ denotes a
\emph{finite chain ring with invariants $(q,s)$}, that is
all its ideals form a chain under inclusion $\texttt{R}\supsetneq
\texttt{R}\theta \supsetneq \cdots\supsetneq
\texttt{R}\theta^{s-1} \supsetneq \texttt{R}\theta^s =\{0\}$, where
$\theta$ a generator of its maximal ideal
$\mathfrak m=\texttt{R}\theta$ and $\texttt{R}/\mathfrak m=\mathbb F_q$, a finite field with $q=p^n$ elements, $p$ is a prime, $n$ a positive number.
The canonical projection $\pi :\texttt{R}\rightarrow \mathbb{F}_q$ naturally extends to
$\texttt{R}[X],$ acting on  the
coefficients.
 We say that the ring $\texttt{S}$ is an
\emph{extension} of $\texttt{R}$ and we denote it by
$\texttt{S}|\texttt{R}$ if $\texttt{R}$ is a subring of $\texttt{S}$
and $1_\texttt{R} = 1_\texttt{S}.$ Let $f$ be a polynomial
over $\texttt{R}$. Then $(f)$ is an ideal of
$\texttt{R}[X]$ generated by $f.$ A polynomial $f$ is called \emph{basic
{irreducible}} if $\pi(f)$ is
{irreducible} over $\mathbb{F}_q.$

 A finite ring $\texttt{S}$ is a \emph{Galois
extension} of $\texttt{R}$ of degree $m$ if
$\texttt{S}\simeq\texttt{R}[X]/(f)$,  where $f$ is a monic basic
irreducible over $\texttt{R}$ of degree $m.$ By
$\texttt{Aut}_\texttt{R}(\texttt{S})$ we denote the group of
ring automorphisms of $\texttt{S}$ which fix the elements of
$\texttt{R}.$ It is well {known}
that if  $\texttt{S}$ is the Galois extension of $\texttt{R}$ of
degree $m$ then $\mathfrak m_S=\texttt{S}\theta$ and
     $\texttt{Aut}_\texttt{R}(\texttt{S})$ is the cyclic group of order $m.$
The map
$\texttt{Tr}_\texttt{R}^\texttt{S}:=\sum\limits_{\rho\in\texttt{Aut}_\texttt{R}(\texttt{S})}\rho$
is called the \emph{trace map} of the Galois extension
$\texttt{S}|\texttt{R}.$ When $\texttt{S}=\mathbb{F}_{q^m},$ we
write $\texttt{T}_m:\mathbb{F}_{q^m}\rightarrow\mathbb{F}_{q}$.
Note that the trace map of $\texttt{S}$ over $\texttt{R}$ is an
$\texttt{R}$-module epimorphism of $\texttt{S}$ and
{$\texttt{T}_m\circ\widetilde{\pi}=\pi\circ\texttt{Tr}_\texttt{R}^\texttt{S},$
where $\widetilde{\pi}:\texttt{S}\rightarrow\mathbb{F}_{q^m}$ is a
natural map with
$\widetilde{\pi}_{\upharpoonright_{\texttt{R}}}=\pi.$}



An $\texttt{R}$-linear code  $\mathcal{C}$ of length $\ell$ is a
$\texttt{R}$-submodule of $\texttt{R}^\ell.$ A linear subcode
$\mathcal{D}$ of  $\mathcal{C}$  is an $\texttt{R}$-submodule of $\mathcal{C}$. $\mathcal{D}$ is said be \emph{proper} if $\mathcal{D}\neq\mathcal{C}$ and
$\mathcal{D}\neq\{\textbf{0}\}.$   The code $\mathcal{C}$ of length $\ell$ and
 rank $k$ is said to be \emph{free} if $\mathcal{C}\simeq\texttt{R}^k$ as
an $\texttt{R}$-module. The following result allows us to count of
all free subcodes of arbitrary free $\texttt{R}$-linear code  with
given  rank.

\begin{Theorem}\label{subfree} If $\mathcal{C}$ is a free
$\texttt{R}$-linear code of length $\ell$ and rank $k$ then the
number of free $\texttt{R}$-linear subcodes of $\mathcal{C}$ of
rank $k',$ is given by $$\left[\left|
\begin{array}{c}
  k \\
  k'
\end{array}%
\right|\right]_{(q,s)}:=q^{(s-1)(k-k')k'}\left[
\begin{array}{c}
   k \\
   k'
\end{array}\right]_q,
\text{ where }\left[\begin{array}{c}
   k \\
   k'  \\
\end{array}\right]_q:=\left\{%
\begin{array}{ll}
    0, & \hbox{if $k<k';$} \\
    1, & \hbox{if $k'=0;$} \\
    \overset{k'-1}{\underset{i=0}{\prod }}\frac{q^k-q^i}{q^{k'}-q^i}, & \hbox{otherwise.}
\end{array}%
\right.$$
\end{Theorem}

\begin{proof} Use \cite[Theorem 2.4.]{HL00}.
\end{proof}



Let $\mathcal{B}$  be a linear code over $\texttt{S}$. We define the \emph{trace code} of $\mathcal{B}$ over $\texttt{R}$ as $$\texttt{Tr}_\texttt{R}^\texttt{S}(\mathcal{B}) =\left\{\left(\texttt{Tr}_\texttt{R}^\texttt{S}(c_0),\texttt{Tr}_\texttt{R}^\texttt{S}(c_1),\cdots,\texttt{Tr}_\texttt{R}^\texttt{S}(c_{\ell-1})\right)\mid(c_0,c_1,\cdots,c_{\ell-1})\in\mathcal{B}\right\},$$
 and the \emph{subring subcode} of $\mathcal{B}$ over $\texttt{R}$ as
     $\texttt{Res}_\texttt{R}(\mathcal{B})=\mathcal{B}\cap\texttt{R}^{\ell}$.
On the other hand, given a linear code  $\mathcal{C}$  over $\texttt{R}$,   the \emph{extension code} of $\mathcal{C}$ to $\texttt{S}$ is the
$\texttt{S}$-submodule  formed by taking all
$\texttt{S}$-linear combinations of elements in $\mathcal{C}.$

Let $\mathcal{B}$ be a free $\texttt{S}$-linear code such that $\sigma(\mathcal{B})\neq\mathcal{B}$ (i.e. it is not Galois invariant, see  \cite{MNR13}). We define the code  $\mathcal{B}_0$ as $\mathcal{B}_0=\texttt{Ext}_\texttt{S}(\texttt{Res}_\texttt{R}(\mathcal{B})).$
 It is clear that $\mathcal{B}_0$ is free as $\mathcal{B}$ is free and from \cite[Corollary~1]{MNR13} it follows that  $\mathcal{B}$ is the non-empty largest Galois-invariant subcode (G-core subcode) of $\mathcal{B}$. The following Lemma provides us a decomposition of a non Galois-invariant free code.

\begin{Theorem}\label{maindeco}
 Let $\mathcal{B}$ be a non Galois-invariant free code over $\texttt{S}$ and $\mathcal{B}_0=\texttt{Ext}_\texttt{S}(\texttt{Res}_\texttt{R}(\mathcal{B}))$,
 then  $\mathcal{B}=\mathcal{B}_0\oplus \mathcal{B}_1$, where
 $\mathcal{B}_1$ has the
 property $\texttt{Res}_\texttt{R}(\mathcal{B}_1)=\{\textbf{0}\}.$
 \end{Theorem}

\begin{proof} Let $k=\texttt{rank}_\texttt{S}(\mathcal{B})$
and $k'=\texttt{rank}_\texttt{S}(\mathcal{B}_0)$. Since
$\sigma(\mathcal{B})\neq\mathcal{B}$, $k>k'$. Also, since $\mathcal{B}_0$ is free as an
$\texttt{S}$-module,  there exists a free $\texttt{S}$-basis
$\{\textbf{x}_1,\textbf{x}_2,\cdots,\textbf{x}_{k'}\}$ of
$\mathcal{B}_0$. Then the set
$\mathcal{B}\backslash\mathcal{B}_0+\mathfrak m^\ell$
is non-empty as $k'<k.$ From \cite[Lemma 3.2]{DKK09}, for any $\textbf{y}_1\in \mathcal{B}\backslash\mathcal{B}_0+\mathfrak m^\ell$, the set 
$\{\textbf{x}_1,\textbf{x}_2,\cdots,\textbf{x}_{k'},\textbf{y}_1\}$
is $\texttt{S}$-linearly independent.  If $k-k'=1$, then we are done and $\mathcal{B}=\mathcal{B}_0\oplus\mathcal{B}_1$, where $\mathcal{B}_1=\left\langle\textbf{y}_1 \right \rangle$. Otherwise, we choose $\textbf{y}_2 \in\mathcal{B}\backslash\left\langle\textbf{x}_1,\textbf{x}_2,\cdots,\textbf{x}_{k'},\textbf{y}_1\right\rangle_\texttt{S}+\mathfrak m^\ell$  and by proceeding in the same manner finally we get $\mathcal{B}=\mathcal{B}_0\oplus\mathcal{B}_1$, where $\mathcal{B}_1=\left\langle\textbf{y}_1,\textbf{y}_2,\cdots,\textbf{y}_{k-k'}\right\rangle_\texttt{S}$ is a free $\texttt{S}$-linear code. Now, if
$\textbf{c}\in\texttt{Res}_\texttt{R}(\mathcal{B}_1)\subseteq\texttt{Res}_\texttt{R}(\mathcal{B})\subseteq\mathcal{B}_0$,
then $\textbf{c}=\textbf{0}$ as $\mathcal{B}_0\cap\mathcal{B}_1=\{\textbf{0}\}.$
\end{proof}

\section{Non Galois-invariant codes}


 Let $\mathcal{L}_\ell(\texttt{S})$ be the set of all
the $\texttt{S}$-linear codes of length $\ell$, we introduce  the
   following set
\begin{align}\label{eq}\mathcal{E}_\texttt{R}(\ell,m)=\left\{\mathcal{B}\in\mathcal{L}_\ell(\texttt{S})\mid \texttt{Tr}_\texttt{R}^\texttt{S}(\mathcal{B})=\texttt{R}^\ell\right\}.\end{align}
Note we have that
$\texttt{S}^\ell\in\mathcal{E}_\texttt{R}(\ell,m)$ and by
Delsarte's theorem \cite{MNR13}, $\mathcal{E}_\texttt{R}(\ell,m)=\left\{\mathcal{B}\in\mathcal{L}_\ell(\texttt{S})\mid
\texttt{Res}_\texttt{S}(\mathcal{B}^\perp)=\{\textbf{0}\}\right\}.$
A criterion for checking whether $\mathcal{B}$ belongs to
$\mathcal{E}_\texttt{R}(\ell,m)$ or not depending on the field
associated to $\texttt{R}$ is given in the following result.

\begin{Proposition}\label{prop:field}  Let $\mathcal{B}$ be an $\texttt{S}$-linear code
then $\mathcal{B}\in\mathcal{E}_\texttt{R}(\ell,m)$ if and only if
$\pi(\mathcal{B})\in\mathcal{E}_{\mathbb F_q}(\ell,m)$, where $\pi(\mathcal{B})$ is the set obtained  applying $\pi$ to each element in $\mathcal{B}$  coordinatewise.
\end{Proposition}

\begin{proof}{ We have that $\pi_{\upharpoonright_{\texttt{R}}} \circ\texttt{Tr}_\texttt{R}^\texttt{S}=\texttt{T}_m\circ\pi$ and therefore $\pi_{\upharpoonright_{\texttt{R}}}
(\texttt{Tr}_\texttt{R}^\texttt{S}(\mathcal{B}))=\texttt{T}_m(\pi(\mathcal{B}))$}
holds. Hence
$\texttt{Tr}_\texttt{R}^\texttt{S}(\mathcal{B})=\texttt{R}^\ell$
if and only if
$\texttt{T}_m(\pi(\mathcal{B}))=\mathbb{F}_{q}^\ell.$
\end{proof}

Note that for $\mathcal{B}, \mathcal{D}\in\mathcal{L}_\ell(\texttt{S})$ such that $\mathcal{D}\subseteq \mathcal{B}$, if
$\mathcal{D}\in\mathcal{E}_\texttt{R}(\ell,m)$ then it follows that
$\mathcal{B}\in\mathcal{E}_\texttt{R}(\ell,m).$
This fact  motivates the following definition.

\begin{Definition} An $\texttt{S}$-linear code $\mathcal{B}$ is said to be \emph{minimal} in $\mathcal{E}_\texttt{R}(\ell,m)$
if $\mathcal{B}\in\mathcal{E}_\texttt{R}(\ell,m)$ and
$\mathcal{B}$ has no proper $\texttt{S}$-linear subcode in
$\mathcal{E}_\texttt{R}(\ell,m).$
\end{Definition}
\noindent Note that $\texttt{S}^\ell$
is the minimal $\texttt{S}$-linear code in
$\mathcal{E}_\texttt{R}(\ell,m)$ if and only if $m=1.$

\begin{Theorem}\label{tmin}  If  $\mathcal{B}$  is the minimal $\texttt{S}$-linear code  in $\mathcal{E}_\texttt{R}(\ell,m)$
then $\mathcal{B}$ is free.
\end{Theorem}

\begin{proof} Let $\{\textbf{c}_1,\cdots,\textbf{c}_k\}$ be a
 set of generators
in $\mathcal{B}$ on the form of \cite[Lemma~1]{MNR13}. Consider
$\mathcal{D}$ the subcode of $\mathcal{B}$ generated by those
$\textbf{c}_i$ such that $\pi(\textbf{c}_{i})\neq\textbf{0}$,
$i=1,\ldots , k$. Therefore
$\pi(\mathcal{D})=\pi(\mathcal{B})\in\mathcal{E}_{\mathbb
F_q}(\ell,m)$ by Proposition~\ref{prop:field}  and hence
$\mathcal{D}\in\mathcal{E}_\texttt{R}(\ell,m)$. Thus by the
minimality of $\mathcal{B}$, we have  $\mathcal{B}= \mathcal{D}$,
and therefore  $\mathcal{B}$ is generated by elements in
$(S\setminus\mathfrak m)^\ell$ thus  it is free.
\end{proof}

\begin{Theorem}\label{tmin1}  Let
$\mathcal{B}\in\mathcal{E}_\texttt{R}(\ell,m).$
{Then} $\mathcal{B}$ is minimal if
and only if {$\mathcal{B}$} is free
and
\begin{align}\texttt{rank}_\texttt{S}(\mathcal{B})=\left\lceil\frac{\ell}{m}\right\rceil,\end{align}
where $\left\lceil\frac{\ell}{m}\right\rceil =\texttt{min}\left\{i\in\mathbb{N}\mid i\leq\frac{\ell}{m}\right\}$.
\end{Theorem}

\begin{proof} Suppose that $\mathcal{B}$ is minimal. From Theorem~\ref{tmin} it is free. Let
$\{\textbf{c}_1,\textbf{c}_2,\cdots,\textbf{c}_k\}$ be
{an $\texttt{S}$-basis of
$\mathcal{B}$ and
$\left\{\alpha_0,\alpha_1,\cdots,\alpha_{m-1}\right\}$ a free
$\texttt{R}$-basis of $\texttt{S}.$} The set
$\left\{\texttt{Tr}_\texttt{R}^\texttt{S}(\alpha_i\textbf{c}_j)\mid
0\leq i<m\text{  and }1\leq j\leq k\right\}$ is a generating set
of $\texttt{Tr}_\texttt{R}^\texttt{S}(\mathcal{B}).$ Let
$u\in\{1,2,\cdots,k\},$ and consider $\mathcal{D}_u$ the
$\texttt{S}$-linear subcode of $\mathcal{B},$ generated by
$\{\textbf{c}_1,\textbf{c}_2,\cdots,\textbf{c}_k\}\backslash\{\textbf{c}_u\}.$
Since $\mathcal{B}$ is minimal in $\mathcal{E}_\texttt{R}(\ell,m)$, 
$\texttt{Tr}_\texttt{R}^\texttt{S}(\mathcal{D}_u)\subsetneq\texttt{Tr}_\texttt{R}^\texttt{S}(\mathcal{B})=\texttt{R}^\ell$
and therefore
$\texttt{rank}_\texttt{R}(\texttt{Tr}_\texttt{R}^\texttt{S}(\mathcal{D}_u))\leq
m(k-1)<
\texttt{rank}_\texttt{R}(\texttt{Tr}_\texttt{R}^\texttt{S}(\mathcal{B}))=\ell\leq
mk.$ Hence $k=\left\lceil\frac{\ell}{m}\right\rceil$. The converse
follows {straightforward} by the minimality of the rank.
\end{proof}

From now on we will use the following notations
$$\Bbbk=\left\lceil\frac{\ell}{m}\right\rceil\hbox{ and
}\mathcal{E}_\texttt{R}(\ell,m,k)=\left\{\mathcal{B}\in\mathcal{E}_\texttt{R}(\ell,m)\mid
\mathcal{B}\simeq\texttt{S}^k\,(\text{ as
$\texttt{R}$-module})\right\}\hbox{ for some }k\geq\Bbbk.$$

\begin{align}\label{ens}\mathcal{M}_\texttt{R}(\ell,m,u)=\left\{\sum\limits_{\mathcal{B}\in\textbf{M}}\mathcal{B}\mid\textbf{M}\subseteq\mathcal{E}_\texttt{R}(\ell,m,\Bbbk)
\text{ and }
\sum\limits_{\mathcal{B}\in\textbf{M}}\mathcal{B}\simeq\texttt{S}^u\,(\text{
as $\texttt{R}$-module})\right\},\quad u\hbox{ an integer number}.
\end{align}

\begin{Theorem} \
The cardinal  of the set $\mathcal{E}_\texttt{R}(\ell,m,k)$ is given by
\begin{align}\label{focus-eq}
\aleph_\texttt{R}(\ell,m,k)=\sum\limits_{u=\Bbbk}^k\left((-1)^{u-\Bbbk}\left|\mathcal{M}_\texttt{R}(\ell,m,u)\right|\left[\left|
\begin{array}{c}
  k \\
  u
\end{array}%
\right|\right]_{(q^m,s)}\right).
\end{align}
Moreover
$\aleph_\texttt{R}(\ell,m,k)\leq\aleph_\texttt{R}(\ell,m,\Bbbk)\left[\left|
\begin{array}{c}
  k \\
  \Bbbk
\end{array}
\right|\right]_{(q^m,s)}.$
\end{Theorem}

\begin{proof} For every
$\mathcal{B}\in\mathcal{E}_\texttt{R}(\ell,m,k)$ there exists
$\mathcal{D}\in\mathcal{E}_\texttt{R}(\ell,m,\Bbbk)$ such that
$\mathcal{D}\subseteq\mathcal{B}.$ Thus
$\mathcal{E}_\texttt{R}(\ell,m,k)=\bigcup\limits_{\mathcal{D}\in\mathcal{E}_\texttt{R}(\ell,m,\Bbbk)}\left[\mathcal{D}\;;\;\rightarrow\right)^{(k)}$
where
$\left[\mathcal{D}\;;\;\rightarrow\right)^{(k)}:=\left\{\mathcal{B}\in\mathcal{L}_\ell(\texttt{S})\;;\;\mathcal{B}\simeq\texttt{S}^k\,\text{(
as $\texttt{R}$-module ) and
}\mathcal{D}\subseteq\mathcal{B}\right\}.$ Note  that
$|\left[\mathcal{D}\;{;}\;\rightarrow\right)^{(k)}|=\left[\left|
\begin{array}{c}
  k \\
  u
\end{array}
\right|\right]_{(q^m,s)},$ for every free $\texttt{S}$-linear code
$\mathcal{D}$ of rank $u.$ {By
inclusion-exclusion principle,
$\aleph_{\texttt{R}}(\ell,m,k)=\sum\limits_{\emptyset\neq\textbf{M}\subseteq\mathcal{E}_{\texttt{R}}(\ell,m,\Bbbk)}(-1)^{|\textbf{M}|-1}\left|\bigcap\limits_{\mathcal{D}\in\textbf{M}}\left[\mathcal{D}\;;\;\rightarrow\right)^{(k)}\right|.$
Since
$\bigcap\limits_{\mathcal{D}\in\textbf{M}}\left[\mathcal{D}\;;\;\rightarrow\right)^{(k)}=\left[\sum\limits_{\mathcal{D}\in\textbf{M}}\mathcal{D}\;;\;\rightarrow\right)^{(k)},$
we have
$\aleph_{\texttt{R}}(\ell,m,k)=\sum\limits_{u=\Bbbk}^k(-1)^{u-\Bbbk}\sum\limits_{\mathcal{B}\in\mathcal{M}_{\texttt{R}}(\ell,m,u)}\left|\left[\mathcal{B}\;;\;\rightarrow\right)^{(k)}\right|.$
} Therefore (\ref{focus-eq}) and the last inequality follow.
\end{proof}

\begin{Corollary} Let $\mathcal{C}$ be a free $\texttt{R}$-linear code
of length $\ell$ and rank $k'.$ The number of distinct free
$\texttt{S}$-linear codes $\mathcal{B}$ of length $\ell$ and
 rank $k$ with the property
$\mathcal{C}=\texttt{Res}_\texttt{R}(\mathcal{B})$ is given by
$\aleph_\texttt{R}(\ell,m,\ell-k+k').$
\end{Corollary}

\begin{proof} By Theorem~\ref{maindeco},
$\mathcal{B}=\texttt{Ext}_\texttt{S}(\mathcal{C})\oplus\mathcal{B}_1,$
where $\mathcal{B}_1$ is a {free}
$\texttt{S}$-linear subcode of $\mathcal{B}$
{of rank} $k-k'.$
 Hence
$\mathcal{B}_1^\perp\in\mathcal{E}_\texttt{R}(\ell,m,\ell-k+k')$
and the number of distinct free $\texttt{S}$-linear codes
$\mathcal{B}$ of length $\ell,$ of rank $k,$ with the property
$\mathcal{C}=\texttt{Res}_\texttt{R}(\mathcal{B})$ is
$\aleph_\texttt{R}(\ell,m,\ell-k+k').$
\end{proof}

\begin{Remark} It was proved in \cite[Theorem 6]{Lyle78} in the case of finite fields that  $\aleph_{\mathbb{F}_q}(\ell,m,\ell-k+k')=\left[
\begin{array}{c}
   \ell-k' \\
   \ell-k
\end{array}%
\right]_{q^m}$ which is in general not true. The counterexample  of the same has been presented at the beginning of the paper.
\end{Remark}


\begin{Theorem}\label{corrected} Let $\texttt{S}$ be the Galois extension of $\texttt{R}$ of degree $m$ and $(k,k')$ be a pair of positive integers.
Then the number of free $\texttt{S}$-linear code $\mathcal{B}$ of
length $\ell$ and of rank $k$ with the property
$k'=\texttt{rank}_\texttt{{R}}\left(\texttt{Res}_\texttt{R}(\mathcal{B})\right)$
is given by
$$\Omega_\texttt{R}(\ell,m,k,k')=\aleph_\texttt{R}(\ell,m,\ell-k+k')\left[\left|
\begin{array}{c}
  \ell \\
  k'
\end{array}%
\right|\right]_{(q,s)}.$$
\end{Theorem}

\begin{proof} There are $\left[\left|
\begin{array}{c}
  \ell \\
  k'
\end{array}%
\right|\right]_{(q,s)}$ free $\texttt{R}$-linear codes
$\mathcal{C}$ of length $\ell$ and
 rank $k'.$ For each such
$\mathcal{C}$ there are $\aleph_\texttt{R}(\ell,m,\ell-k+k')$ free
$\texttt{S}$-linear codes $\mathcal{B}$ of length $\ell$ of rank
$k,$ with the property
$\mathcal{C}=\texttt{Res}_\texttt{R}(\mathcal{B})$ and thus the
result follows.
\end{proof}

\section{Counting codes over finite PIRs}

In this section, we extend our results to  linear codes over finite principal ideal rings (PIRs). For a
detailed treatment of the theory of finite principal ideal rings
we refer the reader to \cite{DKK09,McD74}.
Let $\texttt{R}$ be a finite PIR with maximal ideals
$\textgoth{m}_1,\cdots,\textgoth{m}_u$ and $s_1,\cdots,s_u$ their
indices of stability respectively. 
 Clearly
$\texttt{R}_t:=\texttt{R}/\textgoth{m}_t^{s_t}$ is a finite chain
ring  with maximal ideal $\textgoth{m}_t/\textgoth{m}_t^{s_t}.$
Then we have the ring
 homomorphism
\begin{align}\label{Phi}
\begin{array}{cccc}
 \Phi: & \texttt{R} & \rightarrow & \texttt{R}_1\times\cdots\times\texttt{R}_u \\
    & a & \mapsto &
    \left({\Phi_1(a),\cdots,\Phi_u(a)}\right),
\end{array}
\end{align}
{where
$\Phi_t(a):=a+\textgoth{m}_t$ and
$\Phi_t:\texttt{R}\rightarrow\texttt{R}_t$ naturally extends to
$\texttt{R}[X]$ acting on the coefficients.} Since the maximal
ideals $\textgoth{m}_1,\cdots,\textgoth{m}_u$ of $\texttt{R}$ are
pairwise coprime,  the ring homomorphism $\Phi$ is a ring
isomorphism  by the Chinese remainder theorem. Thus $\texttt{R} =
\texttt{CRT}(\texttt{R}_1,\cdots,
\texttt{R}_u)=\Phi^{-1}(\texttt{R}_1\times\cdots\times
\texttt{R}_u)$ and we say that $\texttt{R}$ is the \emph{Chinese
product of rings} $\{\texttt{R}_t\}_{t=1}^u.$
We extend the ring
isomorphism (\ref{Phi}) coordinatewise to an isomorphism of
$\texttt{R}$-modules
\begin{align}
\begin{array}{cccc}
 \Phi: & \texttt{R}^\ell & \rightarrow & \overset{u}{\underset{t=1}{\prod}}\texttt{R}_t^\ell.
\end{array}
\end{align}
If $\mathcal{C}$  is {an} $\texttt{R}$-linear code of length $\ell$ then  $
  \mathcal{\mathcal{C}} =  \texttt{CRT}\left(\mathcal{C}_1,\mathcal{C}_2,\cdots,\mathcal{C}_u\right)$ $=\left\{\Phi^{-1}(c_1,c_2,\cdots,c_u)\mid c_t\in\mathcal{C}_t \right\}$, where
$\mathcal{C}_t=\Phi_t(\mathcal{C})$ is a $\texttt{R}_t$-linear
codes of length $\ell$. We call $\mathcal{C}$ the \emph{Chinese product of codes}
$\mathcal{C}_1,\mathcal{C}_2,\ldots ,\mathcal{C}_u,$ and
$\mathcal{C}_t$ is called $t^{\mbox{th}}$ component of
$\mathcal{C}.$

\begin{Lemma}[{\cite[Theorem 2.4]{DL09}}]\label{crt}
Let $\{\mathcal{C}_t\}_{t=1}^u$ be $\texttt{R}_t$-linear codes of
length $\ell$ and $\mathcal{C}$ the Chinese product them.
{Then}
 $\mathcal{C}$ is a free $\texttt{R}$-linear code if and only if each $\mathcal{C}_t$ is a free $\texttt{R}_t$-linear code with the rank $\texttt{rank}_{\texttt{R}}(\mathcal{C});$
\end{Lemma}

\begin{Corollary}[{\cite[Lemma 2.3]{DHL11}}] Let $\mathcal{C}$ be an $\texttt{R}$-linear code of length $\ell$ and $\mathcal{D}$ an $\texttt{R}$-linear subcode of
$\mathcal{C}$ with rank $k.$ Then $\mathcal{D} =
\texttt{CRT}(\mathcal{D}_1, \mathcal{D}_2, \cdots,
\mathcal{D}_u),$ where $\mathcal{D}_t=\Phi_t(\mathcal{D})$ is an
$\texttt{R}_t$-linear subcode of $\Phi_t(\mathcal{C})$ and
$\texttt{max}\{\texttt{rank}_{\texttt{R}_t}(\mathcal{D}_t)\}=k.$
\end{Corollary}

The following results are the  counterparts for PIR's of the ones in the previous section.

\begin{Theorem} Let $\texttt{R}$ be a principal ideal ring such that $\texttt{R} =
\texttt{CRT}(\texttt{R}_1, \texttt{R}_2, \cdots,\texttt{R}_u)$,
where $\texttt{R}_t$ are chain rings with invariants
$(q_t,s_t),$ and $\mathcal{C}$ be
{a} free $\texttt{R}$-linear code
of rank $k.$ Then the number of free $\texttt{R}$-linear subcode
$\mathcal{D}$ of $\mathcal{C},$ with rank $k'$ is
$$\left[\left|
\begin{array}{c}
  k \\
  k'
\end{array}%
\right|\right]_{\texttt{R}}=\prod_{t=1}^{u}\left[\left|
\begin{array}{c}
 k \\
  k'
\end{array}%
\right|\right]_{(q_t,s_t)}.$$
\end{Theorem}


From \cite[Proposition 1.2(1), pp.80]{DI71} and \cite[Proposition
1.1.]{EW67}, a Galois extension of a finite PIR is defined as follows:

\begin{Definition} Let $\texttt{R}$ be the Chinese product of finite
chain rings
$\{\texttt{R}_t\}_{{t=1}}^{u}$, the
ring $\texttt{S}$ is a Galois extension of $\texttt{R}$ with
Galois group $\texttt{Gal}_\texttt{R}(\texttt{S})$ if and only if
\begin{enumerate}
    \item $\texttt{S}\simeq\texttt{R}[X]/(f),$ {where $\Phi_t(f)$ is monic basic  irreducible polynomial over $\texttt{R}_t$ of degree $|\texttt{Gal}_\texttt{R}
    (\texttt{S})|;$}
    \item $\texttt{Gal}_\texttt{R}(\texttt{S})$ is a cyclic group generated by $\sigma,$ where
$\sigma=\Phi^{-1}\circ(\sigma_1\times\sigma_2\times\cdots\times\sigma_u)\circ\Phi$
and $\sigma_t$ is a generator of
$\texttt{Aut}_{\texttt{R}_t}(\texttt{S}_t),$ with
$\texttt{S}_t=\Phi_t(\texttt{S}).$
 \end{enumerate}
\end{Definition}

The following theorem presets a generalization of a Theorem ~\ref{corrected} to PIRS.
\begin{Theorem} {Let $\texttt{R}$ be the Chinese product of finite
chain rings $\{\texttt{R}_t\}_{t=1}^{u}$, and $\texttt{S}$ be the
Galois extension of $\texttt{R}$ with Galois group
$\texttt{Gal}_\texttt{R}(\texttt{S}).$} Then the number of free
$\texttt{S}$-linear codes $\mathcal{B}$
{of length $\ell$} such that
$k=\texttt{rank}_\texttt{S}(\mathcal{B})$ and
$k'=\texttt{rank}_\texttt{\texttt{{R}}}(\mathcal{B}\cap\texttt{R}^\ell)$
is
$$\widehat{\Omega}_{\texttt{R}}(\ell,m,k,k'):=\prod\limits_{t=1}^{u}\Omega_{\texttt{R}_t}(\ell,m,k,k'),$$
{where
$m=|\texttt{Gal}_\texttt{R}(\texttt{S})|.$}
\end{Theorem}





\bibliography{PIDS}
\bibliographystyle{elsarticle-num}
\end{document}